%% file: main.tex
\documentclass[10pt,a4paper]{article}
\usepackage{a4wide}
\usepackage{graphicx}
\usepackage{amsmath,amssymb}
\usepackage{cleveref}
\usepackage{url}
\usepackage{mathdots}
\usepackage{amsthm}
\newtheorem{theorem}{Theorem}
\newtheorem{definition}[theorem]{Definition}
\newtheorem{lemma}[theorem]{Lemma}
\newtheorem{corollary}[theorem]{Corollary}

\newtheorem{proposition}[theorem]{Proposition}
\newcommand{\lz}{\mathsf{LZ}_\mathsf{77}}
\newcommand{\lze}{\mathsf{LZ}_\mathsf{end}}
\newcommand{\ze}{\mathsf{Z'}}
\newcommand{\lzph}{\mathsf{z}_\mathsf{77}}
\newcommand{\lzeph}{\mathsf{z}_\mathsf{end}}
\newcommand{\zeph}{\mathsf{z'}}
\newcommand{\beg}{\mathsf{b}}
\newcommand{\en}{\mathsf{e}}
\newcommand{\lzelast}[1]{\mathsf{LZ}_\mathsf{end}(#1)\mathsf{.last}}
\newcommand{\zelast}[1]{\mathsf{Z'}_{#1}\mathsf{.last}}
\begin{document}
\title{On the approximation ratio of LZ-End to LZ77}
\author{Takumi Ideue$^{1}$ \quad 
Takuya Mieno$^{2,3}$ \quad 
Mitsuru Funakoshi$^{2,3}$ \quad \\ 
Yuto Nakashima$^{2}$ \quad 
Shunsuke Inenaga$^{2,4}$ \quad 
Masayuki Takeda$^{2}$ \\ 
{$^1$ Department of Information Science and Technology,} \\
{Kyushu University, Fukuoka, Japan}\\
  {\texttt{ideue.takumi.274@s.kyushu-u.ac.jp}}\\
{$^2$ Department of Informatics, Kyushu University, Fukuoka, Japan}\\
  {\texttt{\{takuya.mieno,mitsuru.funakoshi,yuto.nakashima,}}\\
  {\texttt{inenaga,takeda\}@inf.kyushu-u.ac.jp}}\\
{$^3$ Japan Society for the Promotion of Science, Tokyo, Japan}\\
{$^4$ PRESTO, Japan Science and Technology Agency, Kawaguchi, Japan}\\
}
\maketitle
\begin{abstract}
  A family of Lempel-Ziv factorizations is a well-studied string structure.
  The LZ-End factorization is a member of the family 
  that achieved faster extraction of any substrings (Kreft \& Navarro, TCS 2013).
  One of the interests for LZ-End factorizations is 
  the possible difference between the size of LZ-End and LZ77 factorizations.
  They also showed families of strings 
  where the approximation ratio of the number of LZ-End phrases 
  to the number of LZ77 phrases asymptotically approaches 2.
  However, the alphabet size of these strings is unbounded.
  In this paper, we analyze the LZ-End factorization of the period-doubling sequence.
  We also show that the approximation ratio for the period-doubling sequence 
  asymptotically approaches 2 for the binary alphabet.
\end{abstract}

\input{introduction}
\input{preliminaries}
\input{period-doubling}
\input{lower-bound}
\input{conclusions}
\input{acknowledgment}

\clearpage

\bibliographystyle{abbrv}
\bibliography{ref}
\end{document}

%% file: introduction.tex
\section{Introduction}
\setcounter{footnote}{0} 

The \emph{Lempel-Ziv 77 compression} (\emph{LZ77})~\cite{LZ77} is
one of the most successful lossless compression algorithms to date.
On the practical side, LZ77 and its variants have been used as a core
of compression software such as zip, gzip, rar,
and compressed formats such as PNG, JPEG, PDF.
In addition to these real world applications,
compressed self-indexing structures based on LZ77 have been proposed~\cite{DBLP:journals/tcs/DoJSS14,DBLP:conf/lata/GagieGKNP12,DBLP:conf/latin/GagieGKNP14,Karkkainen96lempel-zivparsing}.
An LZ77-based compressed representation of a string
allowing for fast access, rank, and select queries also exists~\cite{DBLP:conf/dcc/BelazzouguiGGKO15}.

On the (more) theoretical side, the left-to-right greedy factorization
in LZ77, a.k.a. the \emph{LZ77-factorization}, has widely been considered for decades.
It parses a given input string $w$ into
a sequence $p_1, \ldots, p_z$ of non-empty substrings
such that $p_1 = w[1]$ and $p_i$ for $i \geq 2$ is the shortest prefix of $p_i \cdots p_z$
that does not occur in $p_1 \cdots p_{i-1}$.
This implies that the prefix $p_{i}[1..|p_{i}|-1]$ occurs in $p_1 \cdots p_{i-1}$,
and such an occurrence is called a \emph{source} of $p_i$\footnote{This version of LZ77 is often called \emph{non-overlapping LZ77} or \emph{LZ77 without self-references}, since each phrase $p_i$ never overlaps with any of its sources.}.

Among many versions of LZ77 (c.f.~\cite{crochemore1981optimal,LZD,DBLP:journals/algorithmica/KosolobovVNP20,DBLP:journals/tcs/KreftN13,DBLP:conf/spire/KuruppuPZ10,LZRR,LZ78}),
this paper focuses on the \emph{LZ-End} compressor
proposed by Kreft and Navarro~\cite{DBLP:journals/tcs/KreftN13}.
It is also based on a greedy parsing $q_1, \ldots, q_{z'}$ of an input string,
with a restriction that for each phrase $q_i$
there has to be a source which ends at the right-end of a phrase in
$q_1, \ldots, q_{i-1}$. 
This constraint permits fast substring extraction
without expanding the whole input string.
It is known that the LZ-End compression can be computed
in linear time in the input string length~\cite{DBLP:conf/esa/KempaK17},
or in compressed space with slight slow-down on compression time~\cite{DBLP:conf/dcc/KempaK17}.

One can regard LZ-End as a mix of LZ77 and LZ78~\cite{LZ78},
since in the LZ78 factorization the source of each phrase
has to begin and end at boundaries of previous phrases.
Since LZ78 belongs to the class of grammar compression~\cite{DBLP:journals/tit/CharikarLLPPSS05},
LZ-End can be seen as a new bridge between grammar compression and LZ77.

Now, a natural question arises.
How good is the compression performance of LZ-End?
Practical evaluation in the literature~\cite{DBLP:journals/tcs/KreftN13} has
revealed that the compression ratio of LZ-End
is quite close to that of LZ77 (at most 20\% worse),
but very little is understood in theory.
As in the literature,
we measure and compare the sizes of LZ-End and LZ77
by the numbers $z'$ and $z$ of their phrases in the factorizations, i.e., ``$z'$ versus $z$''.

Since LZ77 is an optimal greedy unidirectional parsing,
$z' \geq z$ always holds.
Thus we are concerned with \emph{the approximation ratio} of LZ-End to LZ77,
which is defined by $z'/z$.
Kreft and Navarro~\cite{DBLP:journals/tcs/KreftN13}
presented a simple family of strings
for which $z'/z$ is asymptotically $2$ over an alphabet of size
$n/3$, where $n$ is the length of the string.
Kreft and Navarro~\cite{DBLP:journals/tcs/KreftN13} conjectured that the upper bound for $z'/z$ is also $2$,
but to our knowledge no non-trivial upper bound is known.

In this paper, we show that the same lower bound for $z'/z$ can be obtained
on a binary alphabet,
thus significantly reducing the number of distinct characters used
in the analysis from $n/3$ to $2$.
In particular, we prove that $z'/z$ is asymptotically 2 
for the \emph{period-doubling sequences},
an interesting family of recursive strings.
While the LZ77-factorization of the period-doubling sequences
has an obvious structure (Proposition~\ref{prop:lz-ph}),
the LZ-End factorization of the period-doubling sequences
has a non-trivial structure and needs careful analysis
(see our extensive discussions in Section~\ref{sec:ratio} for detail).

Since the LZ77 factorization (without self-references) and
the LZ-End factorization for the unary string $a^n$ are the same,
our result uses a minimum possible number of distinct characters
to achieve such a lower bound for $z'/z$.

\paragraph*{\bf Related work.}
A famous variant of the LZ77 factorization,
which is called the \emph{C-factorization}~\cite{crochemore1981optimal}
and is denoted by $w = c_1 \cdots c_x$,
differs from the LZ77 in that each phrase $c_i$ is either a fresh character or 
the longest prefix of $c_i \cdots c_{x}$ that occurs in $c_1 \cdots c_{i-1}$.
The size $x$ of the C-factorization 
is known to be a lower bound
for the size of the smallest grammar which generates only the input string~\cite{DBLP:journals/tcs/Rytter03}.
A comparison of the LZ77 factorization and the C-factorization 
was also considered in the literature~\cite{DBLP:conf/mfcs/BerstelS06,DBLP:journals/algorithms/MitsuyaNIBT21}.
The structure of the C-factorization of the period-doubling sequences
was investigated in~\cite{DBLP:conf/mfcs/BerstelS06}.
We emphasize that our analysis of the LZ-End factorization
of the period-doubling sequences is independent
and is quite different from this existing work~\cite{DBLP:conf/mfcs/BerstelS06}.

\emph{Relative LZ} (\emph{RLZ}) is a practical modification
of LZ77 which efficiently compresses
a collection of highly repetitive sequences~\cite{DBLP:conf/spire/KuruppuPZ10}.
In~\cite{DBLP:journals/algorithmica/KosolobovVNP20}
an RLZ-based factorization of a string,
called the \emph{ReLZ-factorization}, was proposed.
The approximation ratio of ReLZ to LZ77 was shown to be $\Omega(\log n)$~\cite{DBLP:journals/algorithmica/KosolobovVNP20}, where $n$ denotes the length of the input string.
On the other hand,
in practice ReLZ was larger than LZ77 by at most a factor of two
in all the tested cases in~\cite{DBLP:journals/algorithmica/KosolobovVNP20}.

%% file: preliminaries.tex
\section{Preliminaries} \label{sec:preliminaries}

\subsection{Strings}
Let $\Sigma$ be the binary {\em alphabet}.
An element of $\Sigma^*$ is called a {\em string}.
The length of a string $w$ is denoted by $|w|$.
The empty string $\varepsilon$ is the string of length 0.
Let $\Sigma^+$ be the set of non-empty strings,
i.e., $\Sigma^+ = \Sigma^* \setminus \{\varepsilon \}$.
For a string $w = xyz$, $x$, $y$ and $z$ are called
a \emph{prefix}, \emph{substring}, and \emph{suffix} of $w$, respectively.
They are called a \emph{proper prefix}, a \emph{proper substring}, and a \emph{proper suffix} of $w$
if $x \neq w$, $y \neq w$, and $z \neq w$, respectively.
Further, we say that $w$ has an \emph{internal occurrence} of $y$
if $y$ occurs in $w$ as a proper substring which is neither a prefix nor a suffix.
The $i$-th character of a string $w$ is denoted by $w[i]$, where $1 \leq i \leq |w|$.
For a string $w$ and two integers $1 \leq i \leq j \leq |w|$,
let $w[i..j]$ denote the substring of $w$ that begins at position $i$ and ends at
position $j$. For convenience, let $w[i..j] = \varepsilon$ when $i > j$.
For any $1 \le i \le |w|$, 
$w[i..|w|]\cdot w[1..i-1]$ is called a \emph{cyclic rotation} of $w$.
If a cyclic rotation of $w$ is not equal to $w$, the cyclic rotation is said to be proper. 
For any string $w$, let $w^1 = w$ and let $w^k = ww^{k-1}$ for any integer $k \ge 2$,
i.e., $w^k$ is the $k$-times repetition of $w$.
A string $w$ is said to be \emph{primitive} 
if $w$ cannot be written as $x^k$ for any $x \in \Sigma^*$ and $k \geq 2$.
Let $\overline{c}$ be the opposite character of $c$ in a binary alphabet 
(e.g., $\overline{a} = b, \overline{b} = a$ for alphabet $\{a, b\}$).
For any non-empty binary string $w$, 
$\widehat{w}$ denotes the string $w[1..|w|-1] \cdot \overline{w[|w|]}$.
We sometimes use $\beg(x)$ and $\en(x)$ 
as the beginning position and the ending position of a substring $x$ of a given string $w$,
if the occurrence of $x$ in $w$ is clear from a discussion.

\subsection{Lempel-Ziv factorizations}
We introduce the Lempel-Ziv 77 and LZ-End factorizations.
\begin{definition} [LZ77~\cite{LZ77}\footnote{This definition of LZ77 is different from the original one~\cite{LZ77} (see~\cite{DBLP:journals/tcs/KreftN13} for more information).}]
  The \emph{Lempel-Ziv 77 factorization} (LZ77 factorization for short) of a string $w$
  is the factorization $\lz(w) = p_{1}, \ldots, p_{z}$ of $w$
  such that $p_i[1..|p_i|-1]$ is the longest prefix of $p_i \cdots p_{z}$ which occurs in $p_1 \cdots p_{i-1}$.
  As an exception, the last phrase $p_z$ can be a suffix of $w$ which occurs in $p_1 \cdots p_{z-1}$.
\end{definition}
\begin{definition} [LZ-End~\cite{DBLP:journals/tcs/KreftN13}]
  The \emph{LZ-End factorization} of a string $w$
  is the factorization $\lze(w) = q_{1}, \ldots, q_{z'}$ of $w$
  such that $q_i[1..|q_i|-1]$ is the longest prefix of $q_i \cdots q_{z'}$ 
  which occurs as a suffix of $q_{1} \cdots q_{j}$ for some $j<i$.
  As an exception, the last phrase $q_{z'}$ can be a suffix of $w$ 
  which occurs as a suffix of $q_1 \cdots q_{j}$ for some $j<z'$.
\end{definition}
We refer to each $p_i$ and $q_i$ as an \emph{LZ phrase} and \emph{LZ-End phrase}, respectively.
For each phrase, associated longest substring is called a \emph{source} of the phrase.
$\lzph(w)$ and $\lzeph(w)$ denote the number of the LZ phrases and the LZ-End phrases of a string $w$, respectively.
For each $1 \le i \le \lzeph(w)$, $\lze(w)[i]$ denotes the $i$-th LZ-End phrase of $\lze(w)$.
Let $\lzelast{w}$ be the last LZ-End phrase of a string $w$, i.e., $\lzelast{w} = \lze(w)[\lzeph(w)]$.
Fig.~\ref{fig:example} shows examples of two factorizations.
\begin{figure}[ht]
  \includegraphics[width=\textwidth]{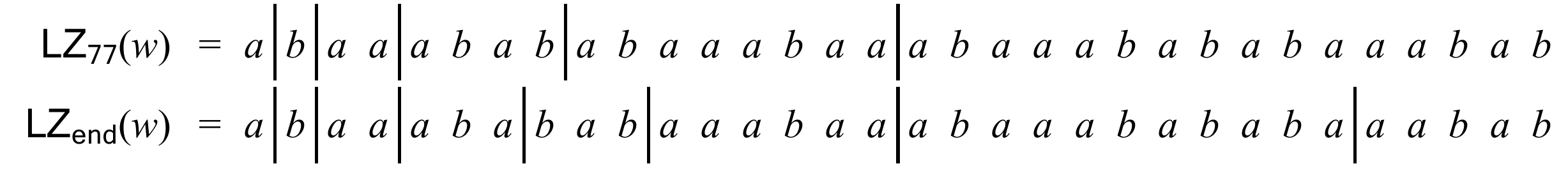}
  \caption{The upper one shows the LZ77 factorization of $w$ 
  and the lower one shows the LZ-End factorization of $w$, 
  where $w = abaaabababaaabaaabaaabababaaabab$.
  This $w$ is the fifth period-doubling sequence $S_5$ which will be defined later.}
  \label{fig:example}
\end{figure}

\subsection{Period-doubling sequence}
The \textit{period-doubling sequence} (cf.~\cite{allouche_shallit_2003}) is 
one of the prominent automatic sequences.
Let $S_k$ be the $k$-th period-doubling sequence for any $k \geq 0$.
The following two definitions are equivalent:
\begin{definition} \label{def:PD-morphism}
  $S_0 = a$ and $S_k = \phi(S_{k-1})$ for $k \ge 1$
  where $\phi$ is the morphism such that $\phi(a) = ab, \phi(b) = aa$.
\end{definition}
\begin{definition} \label{def:PD-doubling}
  $S_0 = a$ and $S_k = S_{k-1} \cdot \widehat{S_{k-1}}$ for $k \ge 1$.
\end{definition}
Let $n_k$ be the length of the $k$-th period-doubling sequence, i.e., $n_k = 2^k$.

%% file: period-doubling.tex
\section{Properties on period-doubling sequence} \label{sec:properties}

The period-doubling sequences have many good combinatorial properties (see cf.~\cite{allouche_shallit_2003}).
In this section, we introduce helpful properties for our results on the period-doubling sequences.

\begin{lemma}\label{lem:primitive}
  For any $k \geq 0$, 
  $S_k$ is primitive.
\end{lemma}
\begin{proof}
  If $S_k$ is not primitive, $S_k$ has a period $2^i$ for some $i$.
  This implies that $S_k[n_{k}/2] = S_k[n_k]$, which contradicts Definition~\ref{def:PD-doubling}.
\end{proof}
\begin{lemma}[Proposition~{8.1.5} of~\cite{lothaire2005applied}] \label{lem:primitive-square}
  If a string $w$ is primitive, $ww$ has no internal occurrence of $w$.
\end{lemma}

\begin{lemma}\label{lem:ABAA}
  For any $k \geq 2$, 
  $S_k = A_k B_k A_k A_k$ where $A_k = S_{k-2}$ and $B_k = \widehat{A_k}$.
  Moreover, $A_k = A_{k-1} B_{k-1}$ and $B_k = A_{k-1} A_{k-1}$
  for any $k \geq 3$.
\end{lemma}
\begin{proof}
  Straightforward from Definition~\ref{def:PD-morphism}.
\end{proof}
\begin{lemma} \label{lem:quart-occ-basis}
  For any $k \geq 2$, 
  $A_k A_k, A_k B_k$, and $B_k A_k$ have no internal occurrence of $A_k$.
  Hence the number of occurrences of $A_k$ in $S_k = A_k B_k A_k A_k$ is $3$.
\end{lemma}
\begin{proof}
  If $k = 2$, the lemma clearly holds. We assume $k \ge 3$.
  Since $A_k = S_{k-2}$, $A_k$ is primitive.
  By Lemma~\ref{lem:primitive-square}, 
  $A_k A_k$ has no internal occurrence of $A_k$.
  Since $A_k B_k = \widehat{A_k A_k}$, $A_k B_k$ also has no internal occurrence of $A_k$.
  Similarly, $A_{k-1} A_{k-1}$ and $A_{k-1} B_{k-1}$ have no internal occurrence of $A_{k-1}$.
  Also, by Lemma~\ref{lem:ABAA}, $B_k A_k$ can be written as $A_{k-1} A_{k-1} A_{k-1} B_{k-1}$.
  These imply that $B_k A_k$ have no internal occurrence of $A_k = A_{k-1} B_{k-1}$.
\end{proof}
\begin{lemma} \label{lem:quart-cyclic-occ}
  For any $k \ge 3$ and any proper cyclic rotation $\alpha$ of $A_k$,
  the number of occurrences of $\alpha$ in $A_k A_k A_k$, $A_k B_k$, and $B_k A_k$ 
  are $2$, $1$, and $0$, respectively.
\end{lemma}
\begin{proof}
  Since $A_k = S_{k-2}$ and Lemma~\ref{lem:primitive}, $A_k$ is primitive.
  This implies that $\alpha$ is also primitive.
  Thus, $A_k A_k$ has exactly one (internal) occurrences of $\alpha$.
  Namely, $\alpha$ occurs in $A_k A_k A_k$ exactly two times.
  Since $A_k B_k = \widehat{A_k A_k}$, $A_k B_k$ also has exactly one (internal) occurrence of $\alpha$.
  Finally, let us consider $B_k A_k = A_{k-1} A_{k-1} A_{k-1} B_{k-1}$.
  In a similar way of the proof of Lemma~\ref{lem:quart-occ-basis},
  we can show that both $A_{k-1} A_{k-1}$ and $A_{k-1} B_{k-1}$ have no internal occurrence of $B_{k-1}$. 
  From this facts and Lemma~\ref{lem:quart-occ-basis},
  $A_{k-1}$ occurs exactly three times and $B_{k-1}$ occurs exactly once in $B_k A_k$.
  If $\alpha = B_{k-1} A_{k-1}$, $\alpha$ cannot occur in $B_k A_k$.
  Otherwise, $\alpha$ can be written as either
  $x B_{k-1} y$ or $x' A_{k-1} y'$
  where $x$ (resp.~$y$) is a non-empty suffix (resp.~prefix) of $A_{k-1}$,
  and $x'$ (resp.~$y'$) is a non-empty suffix (resp.~prefix) of $B_{k-1}$.
  If $\alpha = x B_{k-1} y$, $\alpha$ cannot occur in $B_k A_k$ due to the constraint of $B_{k-1}$.
  If $\alpha = x' A_{k-1} y'$,
  $\alpha$ cannot occur in $B_k A_k$ due to the constraint of $A_{k-1}$
  and the difference between the last characters of $A_{k-1}$ and $x'$.
  Therefore $\alpha$ cannot occur in $B_k A_k$ in all cases.
\end{proof}

%% file: lower-bound.tex
\section{Factorizations of period-doubling sequence} \label{sec:ratio}

By the definition of LZ77, the following proposition immediately holds:
\begin{proposition} \label{prop:lz-ph}
  $\lz(S_k) = (S_0, \widehat{S_0}, \widehat{S_1}, \ldots, \widehat{S_{k-1}})$ and thus $\lzph(S_k) = k+1$.
\end{proposition}
In this section, we mainly discuss the LZ-End factorization of the period-doubling sequence,
and give the following result.
\begin{theorem} \label{thm:lz-end}
  $\lzeph(S_k) = 2k-f(k)$ where $f(k) = O(\log^* k)$.
\end{theorem}

By Proposition~\ref{prop:lz-ph} and Theorem~\ref{thm:lz-end}, we can reach our goal of this paper:
\begin{corollary} \label{cor:ratio}
  There exists a family of \emph{binary} strings $w$ 
  such that the ratio $\lzeph(w)/\lzph(w)$ asymptotically approaches $2$.
\end{corollary}

In the rest of this paper, we show Theorem~\ref{thm:lz-end}.
The next lemma gives the LZ-End factorization of the period-doubling sequence.
Notice that statement (I) in the lemma is not an immediate property for the LZ-End factorization
due to the next example.
Let $S = abaababaabbabbaababa$.
Then,
\begin{eqnarray*}
  \lze(S)    &=& a|b|aa|ba|baab|bab|baabab|a, \\
  \lze(Saba) &=& a|b|aa|ba|baab|bab|baababaaba.
\end{eqnarray*}

\begin{lemma} \label{lem:lzend-pd}
  For any $k \ge 5$, the following statements (I)-(IV) hold.
  \begin{enumerate}
    \renewcommand{\labelenumi}{(\Roman{enumi})}
    \item $\lze(S_k)[i] = \lze(S_{k-1})[i]$ for every $1 \le i \le \lzeph(S_{k-1})-1$.
    \item $\lzeph(S_k) \geq \lzeph(S_{k-1})+1$.
  \end{enumerate}
  Let
  \begin{eqnarray*}
    w_k &=& \lze(S_k)[\lzeph(S_{k-1})], \\
    x_k &=& \lze(S_k)[\lzeph(S_{k-1})+1], \\
    y_k &=& S_k[\en(x_k)+1..n_k]~\text{(possibly empty)}.
  \end{eqnarray*}
  \begin{enumerate}
    \renewcommand{\labelenumi}{(\Roman{enumi})}
    \setcounter{enumi}{2}
    \item
          If $w_k \ne \lzelast{S_{k-1}}$,
          \[
            |w_k| = \frac{1}{8}n_k + 1, |x_k| = \frac{3}{8}n_k, 
            |y_k| = \frac{3}{16}n_{\ell(k)} - (k - \ell(k)) - 1,
          \]
          where $\ell(k) = \max\{i \mid i \leq k, w_i = \lzelast{S_{i-1}}\}$.\\
          Otherwise (if $w_k = \lzelast{S_{k-1}}$),
          \[
            |w_k| = \frac{3}{16}n_k, |x_k| = \frac{5}{16}n_k + 1, |y_k| = \frac{3}{16}n_k - 1.
          \]
    \item If $|y_k| \geq 2$, $y_k[1..|y_k|-1]$ has another occurrence to the left
          which ends with some LZ-End phrase of $S_k$.
          Namely, $y_k$ is the last LZ-End phrase of $S_k$ if $y_k$ is not empty.
  \end{enumerate}
\end{lemma}

\begin{proof}
  In this proof, we use $\ze_k = \lze(S_k)$ and $\zeph_k = \lzeph(S_k)$ for simplicity.
  We prove this lemma by induction on $k$.
  
  Suppose that $k=5$.
  The LZ-End factorizations of $S_4, S_5$ are
  \begin{eqnarray*}
    \ze_4 &=& a|b|aa|aba|bab|aaabaa, \\
    \ze_5 &=& a|b|aa|aba|bab|aaabaa|abaaabababa|aabab.
  \end{eqnarray*}
  Statements (I) and (II) clearly hold.
  Then, $w_5 = aaabaa, x_5 = abaaabababa, y_5 = aabab$.
  Hence, statement (III) holds since $n_5 = 32$ and $w_5 = \zelast{4}$ (i.e., the latter case).
  Statement (IV) also holds since $y_5[1..4] = aaba$ has an occurrence 
  which ends with the fourth phrase $aba$.
  
  Suppose that all the statements hold 
  for any $k \in [5, k'-1]$ for some $k'>5$.
  We show that all the statements hold for $k'$.
  Firstly, suppose on the contrary that statement (I) does not hold for $k'$.
  This implies that there exists a phrase $T = S_{k'}[\beg(\ze_{k-1}[i])..j]$ 
  for some $i < \zeph_{k'-1}$ and $j > n_{k'-1}$.
  Since $|x_{k'-1}y_{k'-1}| \geq \frac{3}{8}n_{k'-1} > \frac{1}{4}n_{k'-1}$ and $x_{k'-1}y_{k'-1}$ is a substring of $T$,
  $T$ has an internal occurrence of the length-$\frac{1}{4}n_{k'-1}$ suffix $A_{k'-1}$ of $S_{k'-1}$.
  By Lemma~\ref{lem:quart-occ-basis} (showing the occurrences of $A_{k-1}$ in $S_{k-1}$), 
  $A_{k'-1}$ occurs exactly three times in $S_{k'}[1..n_{k'-1}]$.
  The first occurrence of $A_{k'-1}$ cannot be included by a source of $T$
  since $A_{k'-1}$ is not a prefix of $T[1..|T|-1]$.
  In addition, the second occurrence of $A_{k'-1}$ also cannot be included by a source of $T$
  since the source overlaps phrase $T$.
  Thus, $T[1..|T|-1]$ cannot have another occurrence to the left as a source of $T$.
  This contradicts that $T$ is an LZ-End phrase of $S_{k'}$ at that position.
  Hence, statement (I) holds for $k'$.
  Due to statement (I), $w_{k'}$ must have $y_{k'-1}$ as a prefix.
  On the other hand, $w_{k'}$ cannot reach the end of $S_{k'}$.
  Hence, statement (II) also holds.
  Thanks to statements (I) and (II) for $k'$, three substrings $w_{k'}$, $x_{k'}$, and $y_{k'}$ are well-defined
  (see Fig.~\ref{fig:lzend-form-mismatch} and~\ref{fig:lzend-form-match} for illustrations).
  \begin{figure}[t]
    \includegraphics[width=\textwidth]{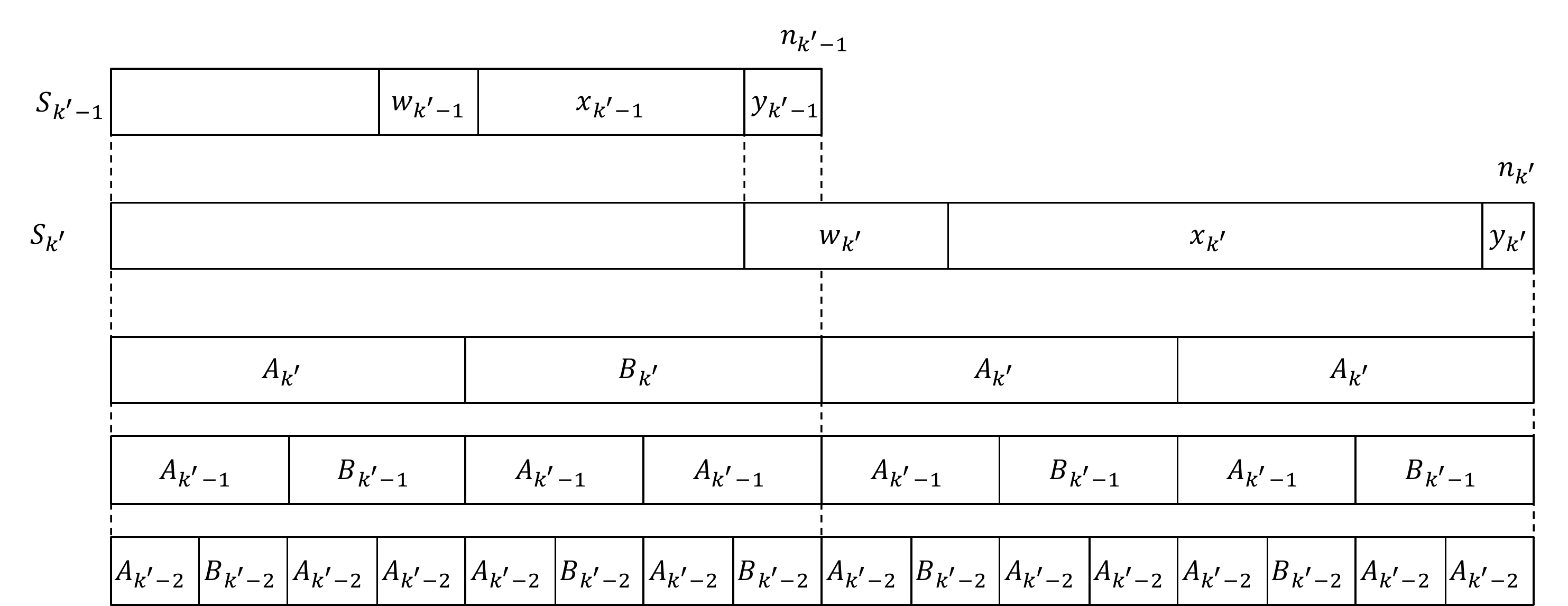}
    \caption{Illustration for the LZ-End factorization when $w_{k'} \ne \zelast{k'-1}$.} 
    \label{fig:lzend-form-mismatch}
  \end{figure}

  Next, we show statements~(III) and (IV).
  \begin{itemize}
  \item Assume that $\ell(k'-1) = \ell(k')$ (i.e., $w_{k'} \ne \zelast{k'-1}$).
        We consider a phrase $w_{k'}$.
        If $|y_{k'-1}| = 0$, $x_{k'-1}$ is the suffix of length $\frac{3}{8}n_{k'-1}$ of $S_{k'-1}$, 
        i.e., $x_{k'-1} = B_{k'-2} A_{k'-1}$.
        From Lemma~\ref{lem:quart-occ-basis}, $x_{k'-1}$ does not have other occurrences to the left.
        This implies that $w_{k'} = x_{k'-1}$.
        This contradicts to $w_{k'} \ne \zelast{k'-1}$.
        Thus, $|y_{k'-1}| > 0$ holds.
        Namely, $x_{k'-1} = \ze_{k'-1}[\zeph_{k'-1}-1]$ and $y_{k'-1} = \zelast{k'-1}$~(see also Fig.~\ref{fig:lzend-form-mismatch}).
        Let $W$ be the string of length $\frac{1}{8}n_{k'}$ which begins at $\beg(\zelast{k'-1})$.
        $\ell(k'-1) = \ell(k')$ also implies that $\ell(k'-1) < k'$.
        Hence, $|y_{k'-1}| < \frac{3}{16}n_{\ell(k'-1)} \leq \frac{3}{32}n_{k'} < \frac{1}{8}n_{k'}$.
        This fact means that $W$ is a proper cyclic rotation of $A_{k'-1}$.
        By Lemma~\ref{lem:quart-cyclic-occ}, $W$ occurs twice to the left 
        (one is in $A_{k'-1} B_{k'-1}$, the other is in $A_{k'-1} A_{k'-1}$).
        Since the second occurrence ends with phrase $\ze_{k'}[\zeph_{k'-1}-1]$,
        $W c_W$ is a candidate of phrase $w_{k'}$ where $c_W$ is the character preceded by $W$.
        Assume on the contrary that a source of phrase $w_{k'}$ is $Wu$ for some $u \in \Sigma^+$
        (see Fig.~\ref{fig:lzend-form-w}).
        \begin{figure}[t]
          \includegraphics[width=\textwidth]{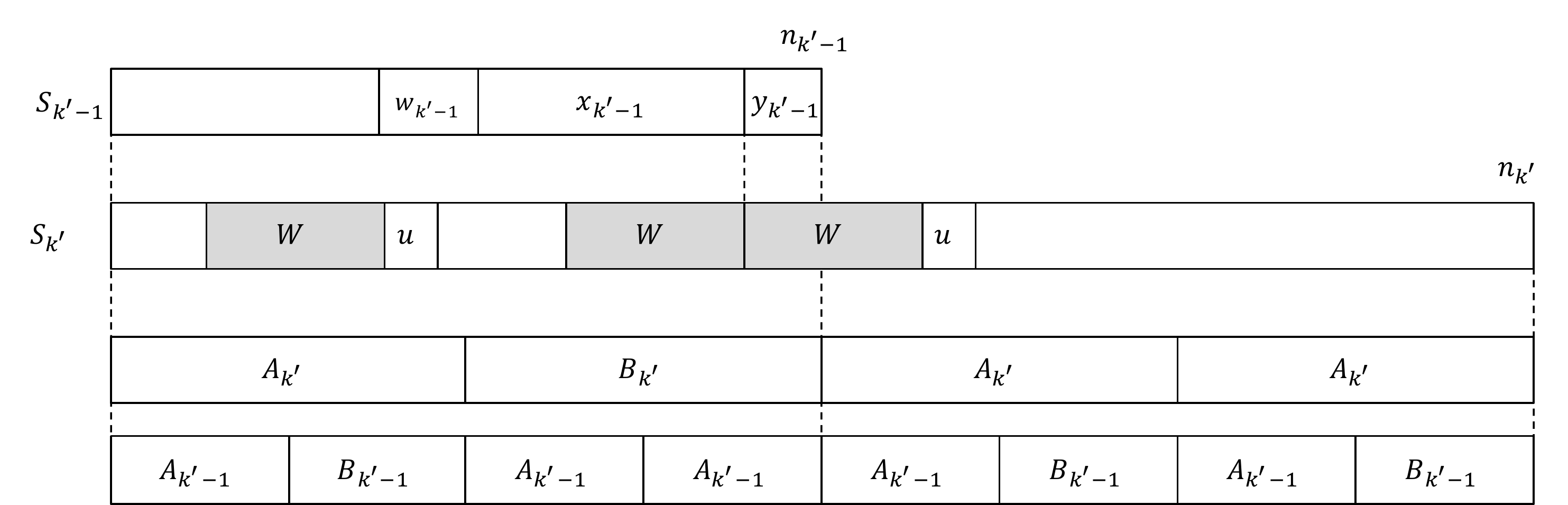}
          \caption{Illustration for a part of the proof.
          $W$ is a candidate of a source of phrase $w_k'$.}
          \label{fig:lzend-form-w}
        \end{figure}
        The second occurrence of $W$ cannot be the beginning position of a source of $w_{k'}$ 
        since $Wu$ overlaps $w_{k'}$.
        Hence, the only candidate of the beginning position of source $Wu$ is in the first $A_{k'-1} B_{k'-1}$.
        Moreover, $Wu$ cannot contain $B_{k'-1}$ since the original $Wu$ occurs in $A_{k'-1} A_{k'-1} \cdots$.
        Thus, $Wu$ is a proper substring of $A_{k'-1} A_{k'-1}$ and $A_{k'-1} B_{k'-1}$.
        In other words, $u'Wu$ is a proper prefix of $A_{k'-1} A_{k'-1}$ and $A_{k'-1} B_{k'-1}$ for some $u'$.
        Since $x_{k'-1}$ is a proper substring of $A_{k'-1} A_{k'-1}$,
        $x_{k'-1}$ also occurs in $u'Wu$.
        Hence, this contradicts that phrase $x_{k'-1}$ ends with $W$
        (i.e., $x_{k'-1}$ has to be a longer phrase.), 
        and then, $w_{k'} = Wc_W$.
        Next, we consider a phrase $x_{k'}$.
        By the definition of the period-doubling sequence, 
        there exists a clear candidate $X$ of a source which ends at
        $\en(x_{k'-1})$~(see Fig.~\ref{fig:lzend-form-x}).
        \begin{figure}[t]
          \includegraphics[width=\textwidth]{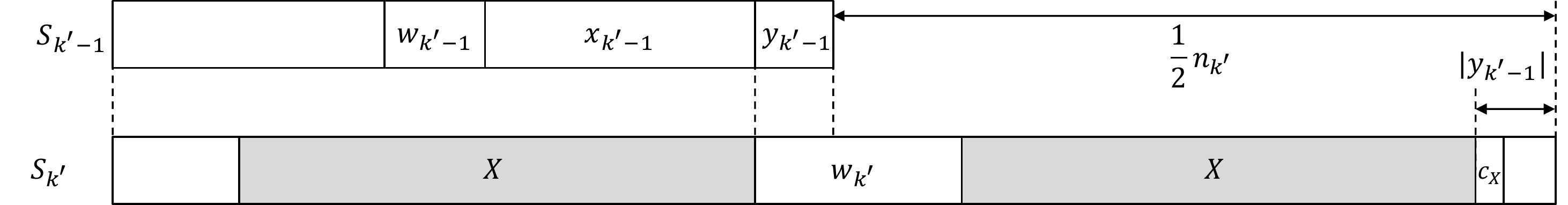}
          \caption{Illustration for a part of the proof.
          $X$ is a candidate of a source of phrase $x_k'$.}
          \label{fig:lzend-form-x}
        \end{figure}
        Then, an equation
        $|y_{k'-1}| + \frac{1}{2}n_{k'} = |w_{k'}| + |X| + |y_{k'-1}|$ stands
        w.r.t. the length of suffix $S_{k'}[\beg(y_{k'-1})..n_{k'}]$.
        Thus, $|X| = \frac{3}{8}n_{k'}-1$ holds 
        since $|w_{k'}| = \frac{1}{8}n_{k'}+1$.
        This implies that $X$ has $B_{k'-1}A_{k'-1}$ as a substring.
        There does not exist a longer candidate 
        since $B_{k'-1}A_{k'-1}$ has only one occurrence to the left.
        Hence, $x_{k'} = Xc_X$ where $c_X$ is the character preceded by $X$.
        Finally, we consider the suffix $y_{k'}$ of $S_{k'}$.
        If $|y_{k'}| \ge 2$,
        from the above discussion, $y_{k'-1}[2..|y_{k'-1}|-1] = y_{k'}[1..|y_{k'}|-1]$ holds.
        Since $y_{k'-1}[2..|y_{k'-1}|-1]$ has an occurrence 
        to the left which ends with some phrase ($\because$ statement (IV) for $k'-1$),
        $y_{k'}[1..|y_{k'}|-1]$ too.
        Therefore, statements (III) and (IV) also hold.
  \item Assume that $\ell(k'-1) \neq \ell(k')$ (i.e., $w_{k'} = \zelast{k'-1}$).
        We can show that all the statements also hold for this case in a similar way.
        If we assume $|y_{k'-1}|>0$, then $|w_{k'}| > |y_{k'-1}|$ holds by the above discussions.
        This contradicts that $w_{k'} = \zelast{k'-1}$, and hence, $|y_{k'-1}|=0$ and $w_{k'} = x_{k'-1}$ hold (see Fig.~\ref{fig:lzend-form-match}).
\begin{figure}[t]
  \includegraphics[width=\textwidth]{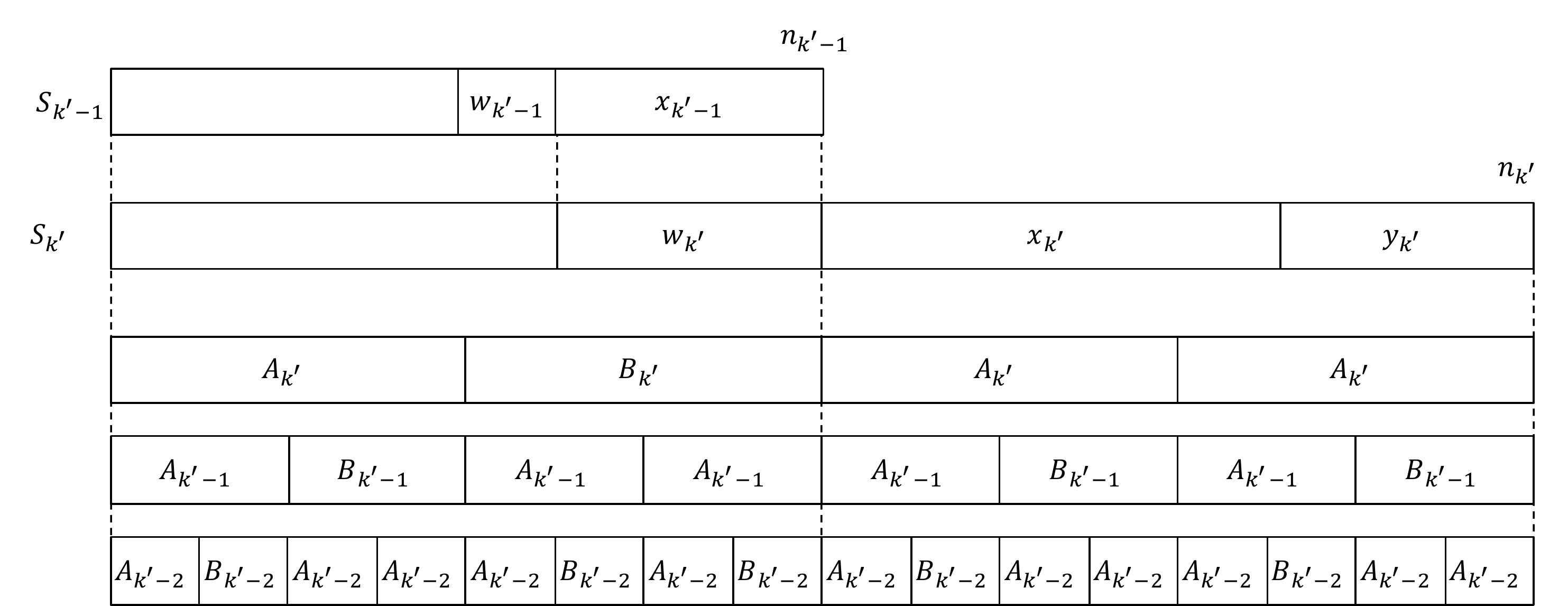}
  \caption{Illustration for the LZ-End factorization when $w_{k'} = \zelast{k'-1}$.} 
  \label{fig:lzend-form-match}
\end{figure}
        Hence, $|w_{k'}| = |x_{k'-1}| = \frac{3}{8}n_{k'-1} = \frac{3}{16}n_{k'}$.
        We consider a phrase $x_{k'}$ that begins at position $\frac{1}{2}n_{k'} + 1$.
        Let $X' = S_{k'}[1.. \en(w_{k'-1})]$ be a clear candidate of a source of $x_{k'}$.
        Since $|X'| = \frac{1}{2}n_{k'} - \frac{3}{16}n_{k'} = \frac{5}{16}n_{k'}$, $X'$ has $A_k'$ as a prefix.
        From Lemma~\ref{lem:quart-occ-basis}, 
        $X'$ is the only candidate of a source, and thus $x_{k'} = X' c_{X'}$ where $c_{X'} = S_{k'}[\frac{13}{16}n_{k'}+1]$ is the character preceded by $X'$.
        Moreover, the length of $y_{k'}$ is $\frac{1}{2}n_{k'} - (\frac{5}{16}n_{k'} + 1) = \frac{3}{16}n_{k'}-1$.
        Since $|y_{k'}| = |w_{k'}| -1$ and phrase $w_{k'}$ is a suffix of $S_{k'-1}$, a source of $w_{k'}$ can be also a source of $y_{k'}$.
        Namely, $y_{k'}$ is the last phrase.
        Thus, all the statements also hold for this case.
  \end{itemize}
  Therefore, this lemma holds.
\end{proof}

We have just finished showing the form of the LZ-End factorization of $S_k$.
Now, we will analyze the number of phrases of the factorization.
Let $\mathcal{K}$ be the sequence of integers $k$ 
which satisfies $\ell(k) = k$.
Let $k_m^*$ denotes the $m$-th smallest integer in $\mathcal{K}$.
Each $k_m^*$ can be represented by the following recurrence formula:
\begin{lemma} \label{lem:re-formula}
  \begin{equation*}
    k_1^* = 5 \text{ and } k_{m}^* = k_{m-1}^* + \frac{3}{16} \cdot 2^{k_{m-1}^*}\text{ for } m \ge 2.
  \end{equation*}
\end{lemma}
\begin{proof}
  Let $m$ be an integer greater than one.
  By the discussion of the proof for the previous lemma,
  $|y_{i-1}|-1 = |y_i|$ holds for any integer $i \in [k_{m-1}^*+1, k_{m}^*-1]$.
  In addition, $|y_{k_{m}^*-1}| = 0$.
  Hence, 
  \[
    k_{m}^* = k_{m-1}^* + |y_{k_{m-1}^*}| + 1 
            = k_{m-1}^* + \frac{3}{16}n_{k_{m-1}^*}
            = k_{m-1}^* + \frac{3}{16} \cdot 2^{k_{m-1}^*}.
  \]
\end{proof}

\begin{lemma} \label{lem:lzend-ph}
  For any $k \geq 5$,
  \begin{equation*}
    \lzeph(S_k) = 2k - f(k),
  \end{equation*}
  where $f(k)$ is a function such that $f(k) = m+1$ if $k \in [k_m^*-1, k_{m+1}^*-2]$.
\end{lemma}
\begin{proof}
  By Lemma~\ref{lem:lzend-pd},
  if $|y_k| = 0$ (i.e., $k+1 \in \mathcal{K}$),
  then $\lzeph(S_k) = \lzeph(S_{k-1}) + 1$ holds,
  otherwise, $\lzeph(S_k) = \lzeph(S_{k-1}) + 2$ holds.
  Hence, for any $k \in [k_m^*-1, k_{m+1}^*-2]$, 
    \[
      \lzeph(S_k) = \lzeph(S_5) + 2(k-5) - (m-1) = 2k - (m+1) = 2k - f(k).
    \]
\end{proof}

\begin{lemma}\label{lem:logstar}
  $f(k) = O(\log^*k)$.
\end{lemma}

\begin{proof}
  By Lemma~\ref{lem:re-formula},
    \[
      k_m^* = O(2^{k_{m-1}^{*}}) \subseteq O\left(2^{2^{\iddots^{2^{k_{1}^*}}}}\right).
    \]
  Thus, $m = O(\log^*k)$ holds.
  This implies that $f(k) = O(\log^*k)$ by Lemma~\ref{lem:lzend-ph}.
\end{proof}
By Lemmas~\ref{lem:lzend-ph} and \ref{lem:logstar}, Theorem~\ref{thm:lz-end} holds.

%% file: conclusions.tex
\section{Conclusions and further work}

Let $z'$ and $z$ be the number of phrases
in the LZ-End and LZ77 factorizations in a string.
In this paper, we proved that the approximation ratio $z'/z$
of LZ-End to LZ77 is asymptotically 2 for the period-doubling sequences.
This significantly reduces the number of distinct characters
needed to achieve such a lower bound from $n/3$ 
(in the existing work~\cite{DBLP:journals/tcs/KreftN13}) to $2$ (in this work).
We believe that our work initiates
analysis of theoretical performance of LZ-End compression.

A lot of interesting further work remains for LZ-End,
including the following:

\begin{itemize}
\item Is our lower bound for the approximation ratio tight? Kreft and Navarro~\cite{DBLP:journals/tcs/KreftN13} conjectured that $z' / z \leq 2$ holds for \emph{any} string. We performed some exhaustive experiments on binary strings and the result supports their conjecture.

\item Is the size $z'$ of the LZ-End factorization a lower bound for the size $g$ of the smallest grammar generating the input string? It is known that the size of the C-factorization~\cite{crochemore1981optimal}, a variant of LZ77, is a lower bound of $g$~\cite{DBLP:journals/tcs/Rytter03,DBLP:journals/tit/CharikarLLPPSS05}. In particular case of the period-doubling sequences, there exists the following small SLP (i.e., grammar in the Chomsky normal form) generating the $k$-th period-doubling sequence:
$S_k  = S_{k-1} T_k$, $T_k  = S_{k-2} S_{k-2}$, \ldots, $S_1  = ab$, $S_0 = a$.
Following~\cite{DBLP:journals/tcs/Rytter03},
the size of an SLP is evaluated by the number of productions
and thus the above grammar is of size $2k+1$.
It is quite close to the size of the LZ-End factorization
which is $2k-O(\log^* k)$ but is slightly larger.

\item Interesting relationships between the size of the C-factorization
  and other string repetitive measures such as
  the size $r$ of the run-length BWT~\cite{Burrows94ablock-sorting},
  the size $s$ of the smallest run-length SLP~\cite{DBLP:journals/dam/NishimotoIIBT20},
  the size $\ell$ of the Lyndon factorization~\cite{10.2307/1970044},
  the size $b$ of the smallest bidirectional scheme~\cite{DBLP:journals/jacm/StorerS82},
  the size $\gamma$ of the smallest string attractor~\cite{DBLP:conf/stoc/KempaP18},
  the substring complexity $\delta$~\cite{DBLP:journals/talg/ChristiansenEKN21},
  have been considered in the literature~\cite{DBLP:journals/jda/BilleGGP18,DBLP:conf/stacs/KarkkainenKNPS17,DBLP:conf/focs/KempaK20,DBLP:conf/latin/KociumakaNP20,DBLP:conf/spire/KutsukakeMNIBT20,DBLP:journals/tit/NavarroOP21,DBLP:conf/cpm/UrabeNIBT19}.
  Can we extend these results to the LZ-End?
\end{itemize}

%% file: acknowledgment.tex
\section*{Acknowledgments}
This work was supported by JSPS KAKENHI Grant Numbers 
JP20J11983 (TM), 
JP20J21147 (MF),
JP18K18002 (YN), JP21K17705 (YN), 
JP18H04098 (MT), JP20H05964 (MT),
and by JST PRESTO Grant Number JPMJPR1922 (SI).